\documentclass[11pt]{article}
\usepackage[left=1in,top=1in,right=1in,bottom=1in]{geometry}
\usepackage{times}

\usepackage{verbatim}
\usepackage{amsmath}
\usepackage{amssymb}
\usepackage{amsthm}
\usepackage{rotating}
\usepackage{algorithm}
\usepackage[noend]{algpseudocode}

\usepackage{silence}
\usepackage{tabularx}
\usepackage{graphicx}
\usepackage{url}
\usepackage{multirow}
\usepackage{subfig}

\usepackage[noend]{algpseudocode}
\usepackage{siunitx}

\newtheorem{lemma}{Lemma}
\newtheorem{proposition}{Proposition}

\allowdisplaybreaks

\begin{document}
\title{Stable Relationships}
\author{Sam Ganzfried\\
Ganzfried Research\\
sam.ganzfried@gmail.com
}

\date{\vspace{-5ex}}

\maketitle

\abstract{We study a dynamic model of the relationship between two people where the states depend on the ``power'' in the relationship. We perform a comprehensive analysis of stability of the system, and determine a set of conditions under which stable relationships are possible. In particular, stable relationships can occur if both people are dominant, but the sum of dominances is below a bound determined by the model's parameters. Stable relationships can also occur if one person is dominant and the other is submissive, provided the level of dominance exceeds the level of submissiveness but not beyond a threshold. We also conclude that a stable relationship is not possible if both people are submissive. While our model is motivated by a social or romantic relationship, it can also be applied to professional or business relationships as well as diplomatic relationships between nations. }

\section{Introduction and model}
\label{se:intro}

We study a dynamic model of the relationship between two people, $A$ and $B.$ The amount that they each ``like'' each other depends on the amount they liked each other at the previous timestep as well as the ``power'' in the relationship. Let $A(t)$ denote the amount that $A$ likes $B$ at time $t$, and $B(t)$ the amount that $B$ likes $A.$  Let $P(t)$ denote the power of player $A$ at time $t$ ($-P(t)$ denotes the power of player $B$).

\begin{eqnarray*}
P(t) &= &\gamma(B(t) - A(t))\\
A(t+1) &= &A(t) + \alpha P(t)\\
B(t+1) &= &B(t) - \beta P(t)
\end{eqnarray*}

We assume that $\gamma > 0$, and that either $\alpha \neq 0$ or $\beta \neq 0$ (if $\alpha = \beta = 0$ then neither player's value depends on the power and the model is trivial). In general we can assume that $\gamma = 1$, and can substitute in $\alpha' = \alpha \gamma$, $\beta' = \beta \gamma$ without affecting our analysis. This will reduce the number of parameters and simplify the model. However, we will choose to keep $\gamma$ in the model, as the model is not very complex and this permits a more intuitive interpretation of the parameters. We will also assume that $A(0) \neq B(0)$, since in that case the system will continue to remain at the starting value.

We can substitute $P(t)$ into the other expressions to obtain the following system:

\begin{equation}
\begin{aligned}
A(t+1) &= A(t) (1 - \alpha \gamma) + \alpha \gamma B(t)\\
B(t+1) &= B(t) (1 - \beta \gamma) + \beta \gamma A(t)
\label{eq:main}
\end{aligned}
\end{equation}

\begin{proposition}
$(A^*,B^*)$ is an equilibrium point of the system defined by Equation~\ref{eq:main} if and only if $A^* = B^*.$
\end{proposition}

\begin{proof}
In an equilibrium, we have that
$$A^* = A^*(1 - \alpha \gamma) + \alpha \gamma B^*$$
$$A^*(1 - 1 + \alpha \gamma) = \alpha \gamma B^*$$
$$\alpha \gamma A^* = \alpha \gamma B^*$$
Since $\gamma > 0$,
$$\alpha A^* = \alpha B^*$$

Similarly, we have
$$\beta B^* = \beta A^*.$$

By assumption, either $\alpha \neq 0$ or $\beta \neq 0$ (or both).
In either case, we have that $A^* = B^*.$
\end{proof}

In order to ascertain the stability of the system, we calculate the eigenvalues of the matrix $M$ corresponding to the system defined by Equation~\ref{eq:main}:

\begin{eqnarray*}
M = 
\begin{bmatrix}
(1 - \alpha \gamma) & \alpha \gamma\\
\beta \gamma & (1 - \beta \gamma) \\
\end{bmatrix}
\end{eqnarray*}

The eigenvalues are
\begin{eqnarray*}
\lambda_1 &= &1\\
\lambda_2 &= &1 - \alpha \gamma - \beta \gamma
\end{eqnarray*}

The corresponding eigenvectors are
\begin{eqnarray*}
v_1 &= &(1,1)\\
v_2 &= &\left(-\frac{\alpha}{\beta}, 1\right)
\end{eqnarray*}

The central concept in analysis of dynamic systems is that of stability~\cite{Luenberger79:Introduction}. A system is \emph{stable} if $|\lambda_L| < 1$, where $\lambda_L$ is the eigenvalue with largest absolute value. If a system is stable, then the dynamics will converge to an equilibrium point. If $|\lambda_L| > 1$, then the system is \emph{unstable}, and will diverge to $\pm \infty.$ If $|\lambda_L| = 1$, then the system is \emph{marginally stable}. Marginal stability is often associated with being a middle ground between these two extremes: ``A marginal system, sometimes referred to as having neutral stability, is between these two types [asymptotically stable and unstable]: when displaced, it does not return to near a common steady state, nor does it go away from where it started without limit''~\cite{Wiki22:Marginal}. It turns out that a marginally stable system can actually exhibit a wide range of behavior such as oscillating between several points, but also converging to an equilibrium (as stable systems) and diverging to infinity (as an unstable system). Essentially, if a system is marginally stable this means that further investigation is needed to determine asymptotic behavior.

\section{Stability characterization}
\label{se:stability}

Since $\lambda_1 = 1$, the system is either marginally stable or unstable depending on whether $|\lambda_2| > 1.$
If $|\lambda_2| > 1$, then the system is unstable, and if $|\lambda_2| \leq 1$ it is marginally stable.

\begin{proposition}
\label{pr:ms}
The system is marginally stable if and only if 
$$-\alpha \leq \beta \leq \frac{2}{\gamma} - \alpha.$$ 
\end{proposition}
\begin{proof}
$\lambda_2 > 1$ if and only if
$$-\alpha \gamma - \beta \gamma + 1 > 1$$
$$\leftrightarrow -\alpha \gamma - \beta \gamma > 0$$
$$\leftrightarrow \beta < -\alpha.$$

$\lambda_2 < -1$ if and only if
$$-\alpha \gamma - \beta \gamma + 1 < -1$$
$$\leftrightarrow -\alpha \gamma - \beta \gamma < -2$$
$$\leftrightarrow \alpha + \beta > \frac{2}{\gamma}$$

So the system is marginally stable if and only if
$$-\alpha \leq \beta \leq \frac{2}{\gamma} - \alpha.$$ 
\end{proof}

Note that $\alpha > 0$ means that player $A$ prefers more power, and $\beta > 0$ means that $B$ prefers more power. Say that a player is \emph{dominant} if they prefer more power, and \emph{submissive} if they prefer less power. The \emph{dominance} of $A$ is $\alpha$, and the \emph{submissiveness} of $A$ is $-\alpha$; similarly, the dominance of $B$ is $\beta$, and the submissiveness of $B$ is $-\beta.$

We now analyze the implications of Proposition~\ref{pr:ms} on the different cases of dominance/submissiveness for the players.

\begin{enumerate}
\item Both $A$ and $B$ are dominant: \\
We have $\alpha \geq 0$, $\beta \geq 0$, and furthermore:
$$-\frac{2}{\gamma} \leq -\alpha \leq 0 \leq \beta \leq \frac{2}{\gamma} - \alpha$$
This is equivalent to the constraint that the sum of the dominances is at most $\frac{2}{\gamma}.$

\item Both $A$ and $B$ are submissive: \\ 
We have $\alpha \leq 0$, $\beta \leq 0$:
$$0 \leq -\alpha \leq \beta \leq \frac{2}{\gamma} - \alpha$$
This can only hold if $\alpha = \beta = 0$, which is excluded by assumption in our model.

\item $A$ is dominant and $B$ is submissive: \\
We have $\alpha \geq 0$, $\beta \leq 0$:
$$-\alpha \leq \beta \leq \frac{2}{\gamma} - \alpha.$$ 
This means that $B$ cannot be more submissive than $A$ is dominant, and $A$ cannot be more dominant than $B$'s submissiveness plus $\frac{2}{\gamma}.$

\item $B$ is dominant and $A$ is submissive: \\
We have $\alpha \leq 0$, $\beta \geq 0$: 
$$0 \leq -\alpha \leq \beta \leq \frac{2}{\gamma} - \alpha$$ 
This means that $A$ cannot be more submissive than $B$ is dominant, and $B$ cannot be more dominant than $A$'s submissiveness plus $\frac{2}{\gamma}.$
\end{enumerate}

In summary, we can achieve a marginally stable system in which both players are dominant, with the sum of dominance bounded by $\frac{2}{\gamma}.$ We can never reach a situation where both players are submissive. We can also reach situations in which $A$ is more (or equally) dominant than $B$ is submissive by at most $\frac{2}{\gamma}$, and situations where $B$ is more (or equally) dominant than $A$ is submissive by at most $\frac{2}{\gamma}.$

\section{Further analysis of marginally stable behavior}
\label{se:ms}
The eigendecomposition of transition matrix $M$ is $M = Q \Lambda Q^{-1}$, where

\begin{equation*}
Q = 
\begin{bmatrix}
1 & -\frac{\alpha}{\beta}\\
1 & 1 \\
\end{bmatrix}
\end{equation*}

\begin{equation*}
\Lambda = 
\begin{bmatrix}
1 & 0\\
0 & 1 - \alpha \gamma - \beta \gamma \\
\end{bmatrix}
\end{equation*}

We have that $M^t = (Q \Lambda Q^{-1})^t = Q \Lambda^t Q^{-1}.$

\begin{equation}
\begin{aligned}
M^t = \frac{1}{\alpha + \beta}
\begin{bmatrix}
\beta + \alpha (1 - \alpha \gamma - \beta \gamma)^t & \alpha - \alpha(1 - \alpha \gamma - \beta \gamma)^t\\
\beta - \beta(1 - \alpha \gamma - \beta \gamma)^t & \alpha + \beta(1 - \alpha \gamma - \beta \gamma)^t\\
\end{bmatrix}
\end{aligned}
\label{eq:matrix}
\end{equation}

Asymptotic behavior of the system is determined by $\lim_{t \rightarrow \infty} M^t.$

\begin{enumerate}
\item $1 - \alpha \gamma - \beta \gamma = 1$: \\
At first glance the system is stable, since we have that $M^t = I.$ However, in this case the matrix $M$ is not actually diagonalizable (the matrix $Q$ is not invertible), so we cannot draw any conclusions based on the eigendecomposition. In Proposition~\ref{pr:unstable}, we show that the system is actually unstable.

\item $0 \leq |1 - \alpha \gamma - \beta \gamma| < 1$: \\
We have

\begin{equation*}
\lim_{t \rightarrow \infty} M^t = \frac{1}{\alpha + \beta}
\begin{bmatrix}
\beta & \alpha\\
\beta & \alpha\\
\end{bmatrix}
\end{equation*}

So we have that 
$$\lim_{t \rightarrow \infty} A(t) = \lim_{t \rightarrow \infty} B(t) 
= \frac{\beta A(0) + \alpha B(0)}{\alpha + \beta}.$$

So the system is stable and converges to the equilibrium point $(A^*,B^*)$ with
$$A^* = B^* = \frac{\beta A(0) + \alpha B(0)}{\alpha + \beta}.$$

\item $1 - \alpha \gamma - \beta \gamma = -1$:\\
For $t$ even, we have

$$\lim_{t \rightarrow \infty} M^t = I.$$
So $\lim_{t \rightarrow \infty} A(t) = A(0)$, $\lim_{t \rightarrow \infty} B(t) = B(0).$

For $t$ odd, we have

\begin{equation*}
\lim_{t \rightarrow \infty} M^t = \frac{1}{\alpha + \beta}
\begin{bmatrix}
\beta - \alpha & 2\alpha\\
2\beta & \alpha - \beta\\
\end{bmatrix}
\end{equation*}

So 
$$\lim_{t \rightarrow \infty} A(t) = \frac{(\beta - \alpha)A(0) + 2 \alpha B(0)}{\alpha + \beta}$$
$$\lim_{t \rightarrow \infty} B(t) = \frac{(\alpha - \beta)B(0) + 2 \beta A(0)}{\alpha + \beta}$$

So the system will oscillate between two points:
$$P_1 = (A_1,B_1) = (A(0), B(0))$$
$$P_2 = (A_2,B_2) = \left(\frac{(\beta - \alpha)A(0) + 2 \alpha B(0)}{\alpha + \beta}, \frac{(\alpha - \beta)B(0) + 2 \beta A(0)}{\alpha + \beta}\right)$$

If $A(0) = B(0)$, then the system trivially just stays at the initial point since it is an equilibrium. Otherwise, the system will oscillate between $P_1$ and $P_2$, neither of which are equilibrium points. However, the average of $P_1$ and $P_2$ is an equilibrium point:
$$(A^*,B^*) = \left(\frac{\beta A(0) + \alpha B(0)}{\alpha + \beta}, \frac{\beta A(0) + \alpha B(0)}{\alpha + \beta}\right)$$
\end{enumerate}

\begin{proposition}
\label{pr:unstable}
If $1 - \alpha \gamma - \beta \gamma = 1$, then the system is unstable.
\end{proposition}
\begin{proof}
Note that the transition matrix is 
\begin{eqnarray*}
M = 
\begin{bmatrix}
(1 - \alpha \gamma) & \alpha \gamma\\
\beta \gamma & (1 - \beta \gamma) \\
\end{bmatrix}
= 
\begin{bmatrix}
(1 - \alpha \gamma) & \alpha \gamma\\
-\alpha \gamma & (1 + \alpha \gamma) \\
\end{bmatrix}
\end{eqnarray*}

Since this matrix is not diagonalizable, we cannot construct the eigendecomposition. However, we can compute
$\lim_{t \rightarrow \infty}M^t$ by instead using the Jordan decomposition. We have $M = S J S^{-1}$ where

\begin{equation*}
S = 
\begin{bmatrix}
1 & -\frac{1}{\alpha \gamma}\\
1 & 0 \\
\end{bmatrix}
\end{equation*}

\begin{equation*}
J = 
\begin{bmatrix}
1 & 1\\
0 & 1 \\
\end{bmatrix}
\end{equation*}

\begin{equation*}
S^{-1} = 
\begin{bmatrix}
0 & 1\\
-\alpha \gamma & \alpha \gamma \\
\end{bmatrix}
\end{equation*}

By Lemma~\ref{le:Jt},
\begin{equation*}
J^t = 
\begin{bmatrix}
1 & t\\
0 & 1 \\
\end{bmatrix}
\end{equation*}

\begin{equation*}
M^t = S J^t S^{-1} = 
\begin{bmatrix}
1 - t \alpha \gamma  & t \alpha \gamma\\
-t \alpha \gamma & 1 + t \alpha \gamma \\
\end{bmatrix}
\end{equation*}

We can see that the system is clearly unstable (for $A(0) \neq B(0)$), with
\begin{eqnarray*}
A(t) &= &A(0) + t \alpha \gamma (B(0) - A(0))\\
B(t) &= &B(0) + t \alpha \gamma (B(0) - A(0)) \\
\end{eqnarray*}

\end{proof}

\begin{lemma}
\label{le:Jt}
For all $t \geq 1$,
\begin{equation*}
J^t = 
\begin{bmatrix}
1 & t\\
0 & 1 \\
\end{bmatrix}
\end{equation*}
\end{lemma}

\begin{proof}
The statement clearly holds for $t = 1.$ Suppose it holds for $t = k$ for some $k \geq 1.$
\begin{equation*}
J^{k+1} = J J^k =
\begin{bmatrix}
1 & 1\\
0 & 1 \\
\end{bmatrix}
\begin{bmatrix}
1 & k\\
0 & 1 \\
\end{bmatrix}
= 
\begin{bmatrix}
1 & k+1\\
0 & 1 \\
\end{bmatrix}
\end{equation*}
\end{proof}

\section{Further analysis of unstable behavior}
\label{se:us}
We have seen several scenarios under which the system is unstable. In this section we explore the asymptotic behavior of these scenarios.

\begin{enumerate}
\item $1 - \alpha \gamma - \beta \gamma = 1$: 
\begin{eqnarray*}
\lim_{t \rightarrow \infty} A(t) &= &sign(\alpha)sign(B(0) - A(0))\infty\\
\lim_{t \rightarrow \infty} B(t) &= &sign(\alpha)sign(B(0) - A(0))\infty = sign(\beta)sign(A(0) - B(0))\infty\\
\end{eqnarray*}
Note that in this case we have $\alpha = - \beta$, so by assumption we do not have $\alpha = \beta = 0.$ So $\alpha \neq 0.$
The asymptotic behavior depends on the sign of $\alpha$ and on the sign of the difference in initial states $B(0) - A(0).$
If $\alpha > 0$ and $B(0) > A(0)$, or if $\alpha < 0$ and $A(0) > B(0)$, then the system will diverge to $(\infty,\infty).$
Otherwise, the system diverges to $(-\infty,-\infty).$ Thus, the system diverges to $(\infty,\infty)$ if the dominant player has a lower initial state value, and diverges to $(-\infty,-\infty)$ if the dominant player
has a higher initial state value.

\item $1 - \alpha \gamma - \beta \gamma > 1$, $\alpha \neq 0$, $\beta \neq 0$: \\
Recall the formula for $M^t$ from Equation~\ref{eq:matrix}. We have:

\begin{equation*}
\lim_{t \rightarrow \infty} M^t = \frac{1}{\alpha + \beta}
\begin{bmatrix}
sign(\alpha)\infty & -sign(\alpha)\infty\\
-sign(\beta)\infty & sign(\beta)\infty\\
\end{bmatrix}
\end{equation*}

Note that $1 - \alpha \gamma - \beta \gamma > 1$ implies that $\alpha + \beta < 0.$
So the denominator is negative. So therefore,

\begin{equation*}
\lim_{t \rightarrow \infty} M^t =
\begin{bmatrix}
-sign(\alpha)\infty & sign(\alpha)\infty\\
sign(\beta)\infty & -sign(\beta)\infty\\
\end{bmatrix}
\end{equation*}

This implies that:

\begin{eqnarray*}
\lim_{t \rightarrow \infty} A(t) &= &sign(\alpha)sign(B(0) - A(0))\infty\\
\lim_{t \rightarrow \infty} B(t) &= &sign(\beta)sign(A(0) - B(0))\infty\\
\end{eqnarray*}

If $\alpha$ and $\beta$ have different signs, then the system will diverge to $(\infty,\infty)$ or $(-\infty,-\infty)$ depending on the sign of $B(0) - A(0).$
If $\alpha$ and $\beta$ are both negative, then the system will diverge to $(\infty,-\infty)$ or $(-\infty,\infty)$ depending on the sign of $B(0) - A(0).$
Note that $\alpha$ and $\beta$ cannot both be positive in this situation since  $\alpha + \beta < 0.$ We are also excluding the case that $\alpha = 0$ or $\beta = 0.$

\item $1 - \alpha \gamma - \beta \gamma > 1$, $\alpha = 0$: \\
We must have $\beta < 0$. We have:

\begin{equation*}
\lim_{t \rightarrow \infty} M^t = \frac{1}{\beta}
\begin{bmatrix}
\beta & 0\\
-sign(\beta)\infty & sign(\beta)\infty\\
\end{bmatrix}
= 
\begin{bmatrix}
1 & 0\\
-\infty & \infty\\
\end{bmatrix}
\end{equation*}

This implies that:

\begin{eqnarray*}
\lim_{t \rightarrow \infty} A(t) &= &A(0)\\
\lim_{t \rightarrow \infty} B(t) &= &sign(B(0) - A(0))\infty\\
\end{eqnarray*}

So the system will diverge to $(A(0),\infty)$ or $(A(0),-\infty)$ depending on the sign of $B(0) - A(0).$

\item $1 - \alpha \gamma - \beta \gamma > 1$, $\beta = 0$: \\
By analogous reasoning to the previous case, we have:

\begin{eqnarray*}
\lim_{t \rightarrow \infty} A(t) &= &sign(A(0) - B(0))\infty\\
\lim_{t \rightarrow \infty} B(t) &= &B(0)\\
\end{eqnarray*}

So the system will diverge to $(\infty,B(0))$ or $(-\infty,B(0))$ depending on the sign of $A(0) - B(0).$

\item $1 - \alpha \gamma - \beta \gamma < -1$, $\alpha \neq 0$, $\beta \neq 0$:\\
From Equation~\ref{eq:matrix}, we see that the behavior will alternate between $\infty$ and $-\infty$ for both
players depending on the parity of $t.$ For even $t$:
 
\begin{equation*}
\lim_{t \rightarrow \infty} M^t =
\begin{bmatrix}
sign(\alpha)\infty & -sign(\alpha)\infty\\
-sign(\beta)\infty & sign(\beta)\infty\\
\end{bmatrix}
\end{equation*}

\begin{eqnarray*}
\lim_{t \rightarrow \infty} A(t) &= &sign(\alpha)sign(A(0) - B(0))\infty\\
\lim_{t \rightarrow \infty} B(t) &= &sign(\beta)sign(B(0) - A(0))\infty\\
\end{eqnarray*}

For odd $t$:

\begin{equation*}
\lim_{t \rightarrow \infty} M^t =
\begin{bmatrix}
-sign(\alpha)\infty & sign(\alpha)\infty\\
sign(\beta)\infty & -sign(\beta)\infty\\
\end{bmatrix}
\end{equation*}

\begin{eqnarray*}
\lim_{t \rightarrow \infty} A(t) &= &sign(\alpha)sign(B(0) - A(0))\infty\\
\lim_{t \rightarrow \infty} B(t) &= &sign(\beta)sign(A(0) - B(0))\infty\\
\end{eqnarray*}

If both $\alpha$ and $\beta$ are positive, then the system will alternate between $(\infty,-\infty)$ and $(-\infty,\infty)$. 
If one of them is positive and the other is negative, the system will alternate between $(\infty,\infty)$ and $(-\infty,-\infty)$. 
Note that we can not have both $\alpha$ and $\beta$ negative under this case, since $1 - \alpha \gamma - \beta \gamma < -1$ implies
that $\alpha + \beta > \frac{2}{\gamma}.$

\item $1 - \alpha \gamma - \beta \gamma < -1$, $\alpha = 0$:\\
We must have $\beta > 0$. For even $t$:

\begin{equation*}
\lim_{t \rightarrow \infty} M^t = \frac{1}{\beta}
\begin{bmatrix}
\beta & 0\\
-sign(\beta)\infty & sign(\beta)\infty\\
\end{bmatrix}
= 
\begin{bmatrix}
1 & 0\\
-\infty & \infty\\
\end{bmatrix}
\end{equation*}

\begin{eqnarray*}
\lim_{t \rightarrow \infty} A(t) &= &A(0)\\
\lim_{t \rightarrow \infty} B(t) &= &sign(B(0) - A(0))\infty\\
\end{eqnarray*}

For odd $t$:

\begin{equation*}
\lim_{t \rightarrow \infty} M^t = \frac{1}{\beta}
\begin{bmatrix}
\beta & 0\\
sign(\beta)\infty & -sign(\beta)\infty\\
\end{bmatrix}
= 
\begin{bmatrix}
1 & 0\\
\infty & -\infty\\
\end{bmatrix}
\end{equation*}

\begin{eqnarray*}
\lim_{t \rightarrow \infty} A(t) &= &A(0)\\
\lim_{t \rightarrow \infty} B(t) &= &sign(A(0) - B(0))\infty\\
\end{eqnarray*}

So the system will alternate between $(A(0),\infty)$ and $(A(0),-\infty).$

\item $1 - \alpha \gamma - \beta \gamma < -1$, $\beta = 0$:\\
By analogous reasoning to case 6, the system will alternate between $(\infty,B(0))$ and $(-\infty,B(0)).$
\end{enumerate}

\section{Conclusions}
\label{se:conclusions}
We have seen that the system can exhibit a wide range of behavior depending on the value of $\lambda_2 = 1 - \alpha \gamma - \beta \gamma$ as well as the signs of $\alpha$, $\beta$, and $B(0) - A(0).$ While analysis of the eigenvalues shows that the system is either marginally stable or unstable, this does not tell the full story. Within the set of marginally stable scenarios we have shown that the system can be stable, unstable, or oscillatory. In particular, we show that stable relationships are possible under our dynamics under certain sets of conditions. The first is that both people are dominant, but not too dominant. They do not need to be equally dominant, but the sum of the dominances must be strictly below $\frac{2}{\gamma}.$ The second is that one person is dominant and the other is submissive, and the dominance of the dominant person strictly exceeds the submissiveness of the submissive person by less than $\frac{2}{\gamma}.$ Note that the magnitudes of the dominance and submissiveness must be similar, but both can be very small or large in absolute terms. Finally, if the sum of the dominances exactly equals $\frac{2}{\gamma}$, then the system will oscillate between two points whose average is an equilibrium.

Interestingly, while it is possible to have a stable relationship between two people who are both dominant, it is not possible if they are both submissive. It is necessary that at least one person is dominant, and that the amount of dominance outweighs the amount of submissiveness of the other. Relationships that do not satisfy the conditions described in the preceding paragraph are unstable, and the state of at least one person will diverge to $\pm \infty.$ One special case is if the sum of the dominances is exactly zero. This situation may appear to be a ``perfect match,'' where one player is exactly as dominant as the other is submissive. We have shown that 
in this situation the system will diverge either to $(\infty,\infty)$ or $(-\infty,-\infty)$, depending on the initial conditions. The $(\infty,\infty)$ case actually seems to represent ``too perfect'' of a match. 

We have also seen that the system can diverge in various ways, including both players going to $\infty$, both going to $-\infty$, one going to $\infty$ and the other going to $-\infty$, and both alternating between $\infty$ and $-\infty.$ In general the system is not robust to changes in the initial conditions. If $A(0)$ exceeds $B(0)$ by a small amount $\epsilon$, the system may behave drastically differently than if $B(0)$ exceeds $A(0)$ by $\epsilon.$

While our model is motivated by social or romantic relationships, it can also be applied to professional or business relationships as well as diplomatic relationships between nations. In all of these settings it is natural to assume that one party will obtain more ``power'' if the other ``likes'' or is reliant on them more. The model could also apply to certain biological interactions between organisms and between automated agents or robots. 

\section{Related research}
\label{se:related}
Felmlee and Greenberg~\cite{Felmlee99:Dynamic} consider a model described by the following system of differential equations. In the system, $x^*$ denotes the goals, preferences, or ideals of the wife, and $y^*$ for the husband, while $x_t$ and $y_t$ refer to the frequency of some behavior of the wife and husband at time $t.$ 
\begin{eqnarray*}
\dot{x} &= & \frac{dx}{dt} = a_1(x^*-x_t) + a_2(y_t - x_t)\\
\dot{y} &= & \frac{dy}{dt} = b_1(y^*-y_t) + b_2(x_t - y_t)\\
\end{eqnarray*}

The second term in each equation resembles our system, though this model differs from ours in several ways. First, their model is continuous-time while ours is discrete. Their model also has 4 parameters, while ours has 3 (and only 2 if we normalize by assuming $\gamma = 1$). 
 They conclude that the stability of the system depends on the values of the four parameters; however, they do not perform exhaustive analysis over all settings. Our analysis allows us to obtain clear intuition for precisely when relationships are stable.

Gottman et al.~\cite{Gottman05:Mathematics} consider the following general discrete model. $W_t$ and $H_t$ denote the wife and husband's ``behavior scores'' at time $t$. The model assumes a sequential structure in which the wife moves first and the husband can then condition his behavior on the wife's action at the current timestep. $I_{HW}(H_t)$ denotes the ``influence'' of the husband's state at time $t$ on the wife's state at time $t+1$ when the husband's current state equals $H_t.$ Similarly $I_{WH}(W_{t+1})$ denotes the influence of the wife's state at time $t+1$ on the husband's state at time $t+1$ when the wife's current state equals $W_{t+1}.$ 
\begin{eqnarray*}
W_{t+1} &= & I_{HW}(H_t) + r_1 W_t + a\\
H_{t+1} &= & I_{WH}(W_{t+1}) + r_2 H_t + b\\
\end{eqnarray*}

The primary difference between their models and ours is the sequential assumption that the wife moves first at each round. Our model assumes that both players select their behavior simultaneously at each time step. If we ignore this distinction, then we can view their model as a generalization of ours where the woman plays the role of player $A$, the husband plays the role of player $B$, $r_1 = 1 - \alpha \gamma$, $r_2 = 1 - \beta \gamma$, $I_{HW}(H_t) = \alpha \gamma H_t$, $I_{WH}(W_t) = \beta \gamma W_t$, and $a = b = 0.$ Their analysis is for the general model and does not consider our specific instantiation of the influence functions. 

An earlier system of differential equations that closely resembles our system was proposed by Strogatz~\cite{Strogatz88:Love}. In this system, $r(t)$ denotes Romeo's love/hate for Juliet at time $t$, and $j(t)$ denotes Juliet's love/hate for Romeo at time $t.$ 
\begin{eqnarray*}
\frac{dr}{dt} &= & a_{11} r + a_{12} j\\
\frac{dj}{dt} &= & a_{21} r + a_{22} j \\
\end{eqnarray*}

Aside from the key difference that their model is continuous while ours is discrete, we can view their model as a generalization of ours where Romeo is player $A$, Juliet is player $B$, $a_{11} = 1 - \alpha \gamma,$ $a_{12} = \alpha \gamma,$ $a_{21} = 1-\beta \gamma,$ $a_{22} = \beta \gamma.$ This article just proposes the model and does not conduct any stability analysis. Subsequent work has analyzed the system for the case for two ``identically cautious lovers'' (i.e., $a_{11} = a_{22},$ $a_{12} = a_{21}$), and has left several other special cases as exercises~\cite{Strogatz15:Nonlinear}.

Prior work has also considered more complex models that incorporate additional factors in relationships such as external events, third parties, and delays. For example, in delay models it is assumed that the change in state at time $t$ may depend on states at some time $t - \tau,$ where $\tau$ is a measure of the delay~\cite{Bielczyk13:Dynamical,Matsumoto18:Delay}. Several extensions of the continuous-time linear systems have been studied by Rinaldi et al.~\cite{Rinaldi15:Modeling}, which include complexities such as incorporating nonlinearity, models for secure vs. insecure and unbiased vs. biased individuals, environmental stress, additional emotional dimensions, and love triangles.

While our model is quite simple, it appears to be the only discrete model that considers simultaneous behavior of the agents. We are able to provide a novel complete characterization of the behavior and stability under our model, and provide novel sets of conditions on the dominance and submissiveness under which stable relationships are possible. Since we have completely analyzed stability for all parameter values, it is unnecessary to consider numerical simulations as some prior works have done for more complex models.

\bibliographystyle{plain}
\bibliography{C://FromBackup/Research/refs/dairefs}

\end{document}